\documentclass[conference, a4paper]{IEEEtran}
\usepackage{graphicx}
\usepackage{amsmath, amsthm, amsfonts, amssymb}

\newtheorem{theorem}{Theorem}

\begin{document}

\title{Soft Covering with High Probability}

\author{
\IEEEauthorblockN{Paul Cuff}
\IEEEauthorblockA{Princeton University}
}

\maketitle

\begin{abstract}
Wyner's soft-covering lemma is the central analysis step for achievability proofs of information theoretic security, resolvability, and channel synthesis.  It can also be used for simple achievability proofs in lossy source coding.  This work sharpens the claim of soft-covering by moving away from an expected value analysis.  Instead, a random codebook is shown to achieve the soft-covering phenomenon with high probability.  The probability of failure is super-exponentially small in the block-length, enabling many applications through the union bound.  This work gives bounds for both the exponential decay rate of total variation and the second-order codebook rate for soft covering.
\end{abstract}

\section{Soft Covering}

Soft covering of a distribution by a codebook is a concept that was introduced by Wyner \cite[Theorem~6.3]{wyner-common-info}.  He developed this tool for the purpose of proving achievability in his work on the common information of two random variables.  Coincidentally, the most prevalent current application of soft covering is for security proofs in wiretap channels (e.g. \cite{bloch-laneman13}), which he also introduced that same year in \cite{wyner-wiretap} but apparently did not see how soft covering applied.

We will focus exclusively on the memoryless case, as did Wyner.  Given a channel $Q_{Y|X}$ and an input distribution $Q_X$, let the output distribution be $Q_Y$.  Also, let the $n$-fold memoryless extensions of these be denoted $Q_{Y^n|X^n}$, $Q_{X^n}$, and $Q_{Y^n}$.

Wyner's soft-covering lemma says that the distribution induced by selecting a $X^n$ sequence at random from a codebook of sequences and passing it through the memoryless channel $Q_{Y^n|X^n}$ will be a good approximation of $Q_{Y^n}$ in the limit of large $n$ as long as the codebook is of size greater than $2^{nR}$ where $R > I(X;Y)$.  In fact, the codebook can be chosen quite carelessly---by random codebook construction, drawing each sequence independently from the distribution $Q_{X^n}$.

Some illustrations of this phenomenon can be found in Figures~1-6.  Here we demonstrate the synthesis of a Gaussian distribution by applying additive Gaussian noise to a codebook.  The solid curve in Figures~\ref{fig:1d5points}, \ref{fig:1d5points_optimized}, \ref{fig:1d32points}, and \ref{fig:1d32points_optimized} is the desired output distribution, while the dashed curves are the approximations induced by the codebooks.  Along the bottom of each graphic the codewords themselves are illustrated with an `x.'  In Figures~\ref{fig:2d25points} and \ref{fig:2d1024points} for the 2-dimensional case, only the distributions induced by the codewords are shown in gray-scale.

The signal-to-noise-ratio is 15 for these examples, meaning that the variance of the desired output distribution is 16 times that of the additive noise.  This gives a mutual information of 2 bits.  Accordingly, if the codebook size is $b^n$, where $b>4$ and $n$ is the dimension, then the output distribution should become a very close match as the dimension increases.  We show two cases:  $b=5$ (i.e. rate of $\log 5$ per channel use) and $b=32$ (i.e. rate of 5 bits per channel use).  Both rates are sufficient asymptotically for soft covering.

In Figure~\ref{fig:1d5points} we see that five randomly chosen codewords does a poor job of approximating the desired output distribution.  Figure~\ref{fig:2d25points} shows the 2-dimensional version of the $b=5$ example.  It's still a poor approximation.  On the other hand, the $b=32$ example shown in Figure~\ref{fig:1d32points} and Figure~\ref{fig:2d1024points} begins to look quite good already in two dimensions.  The benefit of the increased dimension seems apparent, as the distribution is able to have a smoother appearance.  This same benefit will ultimately occur in the $b=5$ case as well, but it will require a higher dimension to manifest itself.

One might also like to consider how good the approximation can be if the codebook is chosen carefully rather than at random.  Figure~\ref{fig:1d5points_optimized} and Figure~\ref{fig:1d32points_optimized} show this for the two cases in one dimension.  We see that while five points is not enough for a great approximation, it can do about as well as 32 randomly chosen codewords.  Also, the 32 carefully placed codewords induce an excellent approximation.  The study of the best codebooks was performed under the name ``resolvability'' in \cite{han-verdu}.  Even though careful placement obviously helps in the cases illustrated, the effect wears off in high dimensions, and you still cannot use a codebook rate $R < I(X;Y)$.  Soft covering theorems, as described herein, exclusively consider random codebooks.  The intended applications are coding theorems in communication settings where often random codebook generation is convenient.

\section{Literature}

The soft-covering lemmas in the literature use a distance metric on distributions (commonly total variation or relative entropy) and claim that the distance between the induced distribution $P_{Y^n}$ and the desired distribution $Q_{Y^n}$ vanishes in expectation over the random selection of the set.\footnote{Many of the theorems only claim existence of a good codebook, but all of the proofs use expected value to establish existence.}  In the literature, \cite{han-verdu} studies the fundamental limits of soft-covering as ``resolvability,'' \cite{hayashi06} provides rates of exponential convergence, \cite{cuff13} improves the exponents and extends the framework, \cite{ahlswede-winter02} and \cite[Chapter~16]{wilde-text} refer to soft-covering simply as ``covering'' in the quantum context, \cite{winter05} refers to it as a ``sampling lemma'' and points out that it holds for the stronger metric of relative entropy, and \cite{hou-kramer14} gives a recent direct proof of the relative entropy result.  A covering lemma found in \cite{bennett-etal14} makes a high probability claim similar to this work; however, it is different in that it only applies to what is referred therein as ``unweighted'' channels.  They then use this to make very strong claims about channel synthesis for general channels, even for worst-case channel inputs.  That covering lemma is closely related to this at a high level but technically quite different.

\begin{figure}
    \centering
    \includegraphics[width=0.4\textwidth]{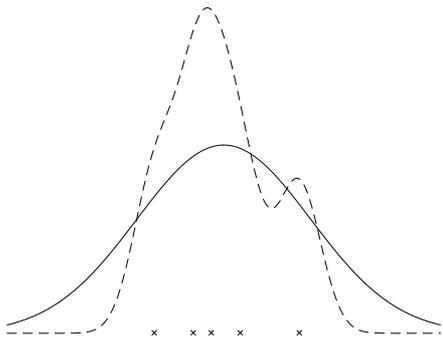}
    \caption{5 randomly selected codewords}
    \label{fig:1d5points}
\end{figure}

\begin{figure}
    \centering
    \includegraphics[width=0.4\textwidth]{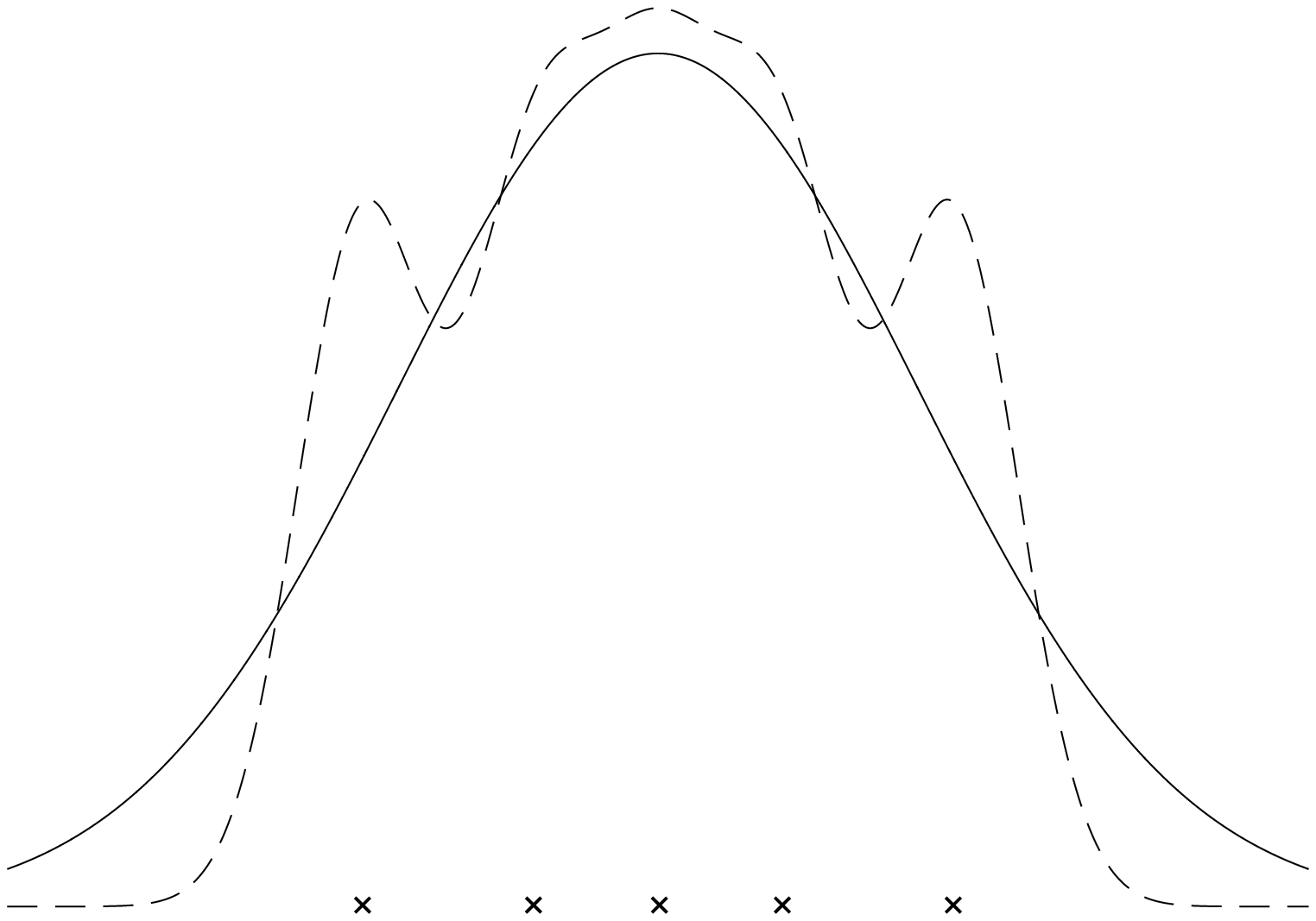}
    \caption{5 carefully selected codewords}
    \label{fig:1d5points_optimized}
\end{figure}

\begin{figure}
    \centering
    \includegraphics[width=0.4\textwidth]{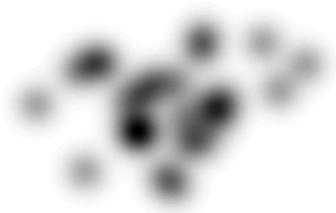}
    \caption{$25 = 5^2$ randomly selected codewords in 2 dimensions}
    \label{fig:2d25points}
\end{figure}

\begin{figure}
    \centering
    \includegraphics[width=0.4\textwidth]{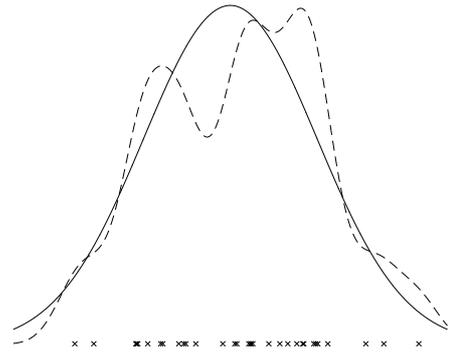}
    \caption{32 randomly selected codewords}
    \label{fig:1d32points}
\end{figure}

\begin{figure}
    \centering
    \includegraphics[width=0.4\textwidth]{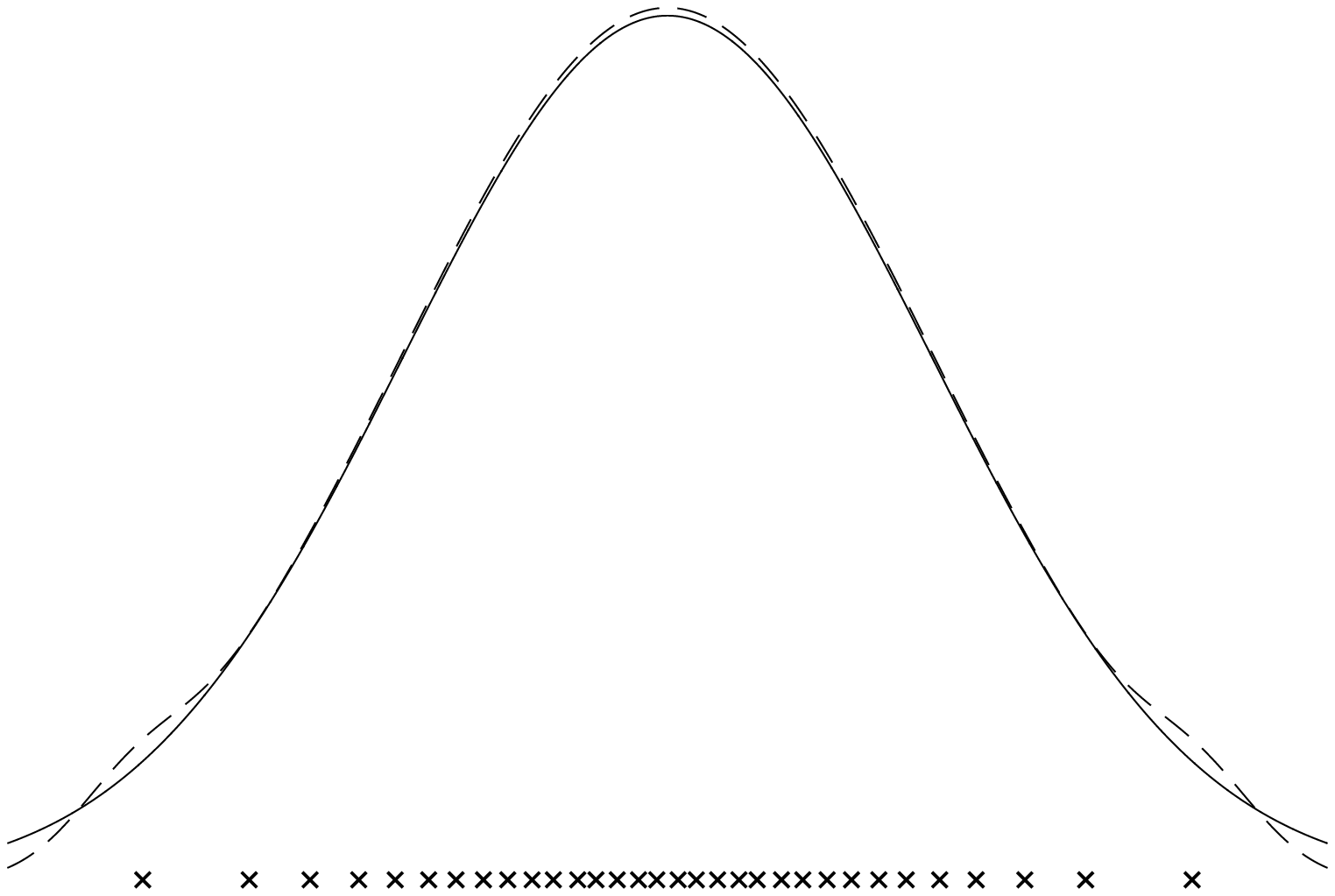}
    \caption{32 carefully selected codewords}
    \label{fig:1d32points_optimized}
\end{figure}

\begin{figure}
    \centering
    \includegraphics[width=0.4\textwidth]{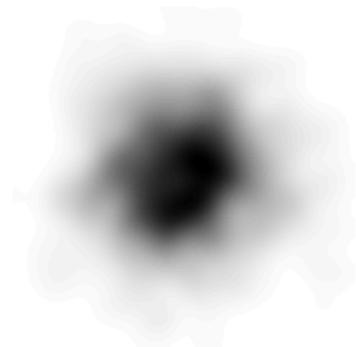}
    \caption{$1024 = 32^2$ randomly selected codewords in 2 dimensions}
    \label{fig:2d1024points}
\end{figure}

Here we give two incarnations of a stronger claim.  First, with high probability with respect to the codebook distribution, for any fixed rate $R > I(X;Y)$, the total variation distance will vanish exponentially quickly with the block-length $n$.  Second, if the codebook rate exceeds $I(X;Y)$ by a vanishing amount of order $\frac{1}{\sqrt{n}}$ (referred to as the second-order rate), then the total variation is bounded by a constant with high probability.  In both cases, the negligible probability of the random codebook not producing these desired results is super-exponentially small.

Both of the results provided in this work have matching results in the literature, but in a weaker form.  For example, the same second-order rate provided in this work was shown in \cite{watanabe-hayashi14} using a bound found in \cite{hayashi06}.  The difference is that the present work shows that a random codebook will achieve this same phenomenon at the same efficient rates with extremely high probability.  Previous results only made claims about the best codebook via expected value arguments.  As we will claim in Section~\ref{section applications}, the high probability aspect of these results is crucial for solving certain problems that were previously challenging.

The results presented in this paper are highly related to results we presented in \cite{cuff15} and \cite{goldfeld-cuff-permuter15}.  In this paper we prove the high probability results directly for the total variation metric (instead of relative entropy), which yields tighter results by a factor of two in the exponent than going through Pinsker's inequality.  Furthermore, we provide second-order rate results not present in any previous publications.

\section{Main Results}

Let us define precisely the induced distribution.  Let $C = \{x^n(m)\}_{m=1}^M$ be the codebook.  Then the induced distribution of $Y^n$, which is a function of the codebook $C$, is the conditional distribution
\begin{align}
    P_{Y^n|C} &= 2^{-nR} \sum_{x^n(m) \in C} Q_{Y^n|X^n=x^n(m)}.
\end{align}

The codebook itself is randomly generated, with the codewords mutually independent and distributed according to
\begin{equation}
    x^n(m) \sim Q_{X^n} \quad \forall m.
\end{equation}
We denote the random codebook with calligraphic text,  as ${\cal C}$.  Thus, $P_{Y^n|C={\cal C}}$ is a random distribution because of the random codebook selection.  For notational brevity, we will refer to this simply as $P_{Y^n|{\cal C}}$.

The size of the codebook is $M = 2^{nR}$.  In the case of Theorem~\ref{theorem second order}, we let the rate vary with $n$ so that it converges down to the asymptotic limit of $I(X;Y)$.

\begin{theorem}[Exponential convergence]
\label{theorem exponential}
    For any $Q_{X}$, $Q_{Y|X}$, and $R > I(X;Y)$, where $X$ and $Y$ have finite supports ${\cal X}$ and ${\cal Y}$, there exists a $\gamma_1 > 0$ and a $\gamma_2 > 0$ such that for $n$ large enough
    \begin{align}
        \mathbb{P} \left( \| P_{Y^n|{\cal C}} - Q_{Y^n}\|_{TV} > e^{-\gamma_1 n} \right) & \leq e^{- e^{\gamma_2 n}},
    \end{align}
    where $\| \cdot \|_{TV}$ is total variation.

    More precisely, for any $n \in \mathbb{N}$ and $\delta \in \big(0,R-I(X;Y)\big)$,
    \begin{equation}
    \mathbb{P} \left( \| P_{Y^n|{\cal C}} - Q_{Y^n}\|_{TV} > 3 \cdot 2^{-n\gamma_{\delta}} \right) \leq \big( 1 + |{\cal Y}|^n \big) e^{-\frac{1}{3} 2^{n\delta}},
    \label{EQ:soft_covering_precise}
    \end{equation}
    where
    \begin{align}
        \gamma_{\delta} &= \sup_{\alpha > 1} \frac{\alpha - 1}{2 \alpha - 1} \big(R - \delta - d_{\alpha} (Q_{X,Y},Q_X Q_Y)\big), \label{EQ:soft_covering_exponent}
    \end{align}
    and $d_{\alpha}(\cdot,\cdot)$ is the R\'{e}nyi divergence of order $\alpha$.
\end{theorem}

\begin{theorem}[Second order rate]
\label{theorem second order}
    For any $Q_{X}$, $Q_{Y|X}$, and $\varepsilon \in (0,1)$, where $X$ and $Y$ have finite supports ${\cal X}$ and ${\cal Y}$, let the rate $R$ vary with $n$ as
    \begin{equation}
        R_n = I(X;Y) + \frac{1}{\sqrt{n}} {\cal Q}^{-1}(\varepsilon) \sqrt{V} + c \frac{ \log n }{n},
    \end{equation}
    where ${\cal Q}$ is one minus the standard normal cdf, $V$ is the variance of $\imath_{X;Y}(X;Y)$, and $c>2$ is arbitrary.  Then for any $d < c-1$ and for $n$ large enough,
    \begin{align}
        \mathbb{P} \left( \| P_{Y^n|{\cal C}} - Q_{Y^n}\|_{TV} > \varepsilon \right) & \leq e^{- n^d}.
    \end{align}
\end{theorem}

\begin{proof}[Proof of Theorem~\ref{theorem exponential}]
    We state the proof in terms of arbitrary distributions (not necessarily discrete).  When needed, we will specialize to the case that ${\cal X}$ and ${\cal Y}$ are finite.
    
    Let the Radon-Nikodym derivative between the induced and desired distributions be denoted as
    \begin{align}
        D_{\cal C}(y^n) &\triangleq \frac{d P_{Y^n|{\cal C}}}{d Q_{Y^n}}(y^n).
    \end{align}
    In the discrete case, this is just a ratio of probability mass functions.
    
    Notice that the total variation of interest, which is a function of the codebook ${\cal C}$, is given by
    \begin{align}
        \| P_{Y^n|{\cal C}} - Q_{Y^n}\|_{TV} &= \frac{1}{2} \int d Q_{Y^n} | D_{\cal C} - 1 | \\
        &= \int d Q_{Y^n} [ D_{\cal C} - 1]_+,
    \end{align}
    where $[z]_+ = \max\{z,0\}$.
    
    Define the jointly-typical set over $x$ and $y$ sequences by
    \begin{align}
        {\cal A}_{\epsilon} &\triangleq \left\{ (x^n, y^n) : \frac{1}{n} \log \frac{d Q_{Y^n|X^n=x^n}}{d Q_{Y^n}} (y^n) \leq I(X;Y) + \epsilon \right\}.
    \end{align}
    
    We split $P_{Y^n|{\cal C}}$ into two parts, making use of the indicator function denoted by $\mathbf{1}$.  Let $\epsilon>0$ be arbitrary, to be determined later.
    \begin{align}
        P_{{\cal C}, 1} &\triangleq 2^{-nR} \sum_{x^n(m) \in {\cal C}} Q_{Y^n|X^n=x^n(m)} \mathbf{1}_{(Y^n,x^n(m)) \in {\cal A}_{\epsilon}}, \\
        P_{{\cal C}, 2} &\triangleq 2^{-nR} \sum_{x^n(m) \in {\cal C}} Q_{Y^n|X^n=x^n(m)} \mathbf{1}_{(Y^n,x^n(m)) \notin {\cal A}_{\epsilon}}.
    \end{align}
    The measures $P_{{\cal C}, 1}$ and $P_{{\cal C}, 2}$ on the space ${\cal Y}^n$ are not probability measures, but $P_{{\cal C}, 1} + P_{{\cal C}, 2} = P_{Y^n|{\cal C}}$ for each codebook ${\cal C}$.
    
    Let us also split $D_{\cal C}$ into two parts:
    \begin{align}
        D_{{\cal C}, 1}(y^n) &\triangleq \frac{d P_{{\cal C}, 1}}{d Q_{Y^n}}(y^n), \\
        D_{{\cal C}, 2}(y^n) &\triangleq \frac{d P_{{\cal C}, 2}}{d Q_{Y^n}}(y^n).
    \end{align}
    
    This allows us also to bound the total variation by a sum of two terms:
    \begin{align}
        \| P_{Y^n|{\cal C}} - Q_{Y^n}\|_{TV} &\leq \int d Q_{Y^n} [ D_{{\cal C},1} - 1]_+ + \int d Q_{Y_n} D_{{\cal C},2} \\
        &= \int d Q_{Y^n} [ D_{{\cal C},1} - 1]_+ + \int d P_{{\cal C},2}. \label{expanded total variation bound}
    \end{align}
    
    Notice that $P_{{\cal C}, 1}$ will usually contain almost all of the probability.  That is, denoting the complement of ${\cal A}_{\epsilon}$ as $\overline{{\cal A}_{\epsilon}}$,
    \begin{align}
        \int d P_{{\cal C}, 2} &= 1 - \int d P_{{\cal C}, 1} \\
        &= 2^{-nR} \sum_{x^n(m) \in {\cal C}} \mathbb{P}_Q \left( \overline{{\cal A}_{\epsilon}} \; \middle| \; X^n = x^n(m, {\cal C}) \right).
    \end{align}
    This is an average of exponentially many i.i.d. random variables bounded between 0 and 1.  Furthermore, the expected value of each one is the exponentially small probability of correlated sequences being atypical:
    \begin{align}
        \mathbb{E} \; \mathbb{P}_Q \left( \overline{{\cal A}_{\epsilon}} \; \middle| \; X^n = x^n(m, {\cal C}) \right) &= \mathbb{P}_Q \left( \overline{{\cal A}_{\epsilon}} \right) \\
        &\leq 2^{-\beta n}, \label{atypical probability bound}
    \end{align}
    where
    \begin{align}
        \beta &= (\alpha - 1) \left( I(X;Y) + \epsilon - d_{\alpha}(Q_{X,Y}, Q_X Q_Y) \right)
    \end{align}
    for any $\alpha>1$, where $d_{\alpha}(\cdot,\cdot)$ is the R\'{e}nyi divergence of order $\alpha$.  Here we use the finiteness of ${\cal X}$ and ${\cal Y}$ to assure that the R\'{e}nyi divergence is finite and continuous for all $\alpha$, which likewise assures a choice of $\alpha$ can be found to give a positive value of $\beta$ if $\epsilon$ is small enough.  We use units of bits for mutual information and R\'{e}nyi divergence to coincide with the base two expression of rate.
    
    Therefore, the Chernoff bound assures that $\int d P_{{\cal C}, 2}$ is exponentially small.  That is,
    \begin{align}
        \mathbb{P} \left( \int d P_{{\cal C}, 2} \geq 2 \cdot 2^{-\beta n} \right) &\leq e^{-\frac{1}{3} 2^{n( R - \beta)}}.
        \label{chernoff for atypical}
    \end{align}
    
    Similarly, $D_{{\cal C}, 1}$ is an average of exponentially many i.i.d. and uniformly bounded functions, each one determined by one sequence in the codebook:
    \begin{align}
        D_{{\cal C}, 1}(y^n) &= 2^{-nR} \sum_{x^n(m) \in {\cal C}} \frac{d Q_{Y^n|X^n=x^n(m)}}{d Q_{Y^n}} (y^n) \mathbf{1}_{(y^n,x^n(m)) \in {\cal A}_{\epsilon}}
    \end{align}
    For every term in the average, the indicator function bounds the value to be between $0$ and $2^{nI(X;Y) + n \epsilon}$.
    The expected value of each term with respect to the codebook is bounded above by one, which is observed by removing the indicator function.
    Therefore, the Chernoff bound assures that $D_{{\cal C},1}$ is exponentially close to one for every $y^n$.  For any $\beta_2>0$:
    \begin{align}
        \mathbb{P} \left( D_{{\cal C},1}(y^n) \geq 1 + 2^{-\beta_2 n} \right) &\leq e^{-\frac{1}{3} 2^{n( R - I(X;Y) - \epsilon - 2 \beta_2 )}} \quad \forall y^n. \label{typical set bound}
    \end{align}
    This use of the Chernoff bound has been used before for a soft-covering lemma in the proof of Lemma~9 of \cite{ahlswede-winter02}.
    
    At this point we will use the fact that ${\cal Y}$ is a finite set.  We use the union bound applied to \eqref{atypical probability bound} and \eqref{typical set bound}, taking advantage of the fact that the space ${\cal Y}^n$ is only exponentially large.  Let ${\cal S}$ be the set of codebooks such that the following are true:
    \begin{align}
    \int d P_{{\cal C},2} &< 2 \cdot 2^{-\beta n}, \\
    D_{{\cal C}, 1}(y^n) &< 1 + 2^{-\beta_2 n} \quad \forall y^n \in {\cal Y}^n.
    \end{align}
    We see that the probability of not being in ${\cal S}$ is doubly exponentially small:
    \begin{align}
    \mathbb{P}({\cal C} \notin {\cal S}) &\leq e^{-\frac{1}{3} 2^{n (R - \beta)}} + |{\cal Y}|^n e^{-\frac{1}{3} 2^{n(R - I(X;Y) - \epsilon - 2 \beta_2)}}.
    \end{align}
    
    What remains is to show that for every codebook in ${\cal S}$, the total variation is exponentially small.  From \eqref{expanded total variation bound} it follows that
    \begin{align}
        \| P_{Y^n|{\cal C}} - Q_{Y^n}\|_{TV} &\leq 2 \cdot 2^{-\beta n} + 2^{-\beta_2 n}.
    \end{align}
    
    Finally, we carefully select $\epsilon$, $\beta_1$ and $\beta_2$ to give the tightest exponential bound, in terms of $\delta$ from the theorem statement:
    \begin{align}
        \epsilon_{\alpha,\delta} &= \frac{ \frac{1}{2} (R-\delta) + (\alpha - 1) d_{\alpha} (Q_{X,Y},Q_X Q_Y) }{ \frac{1}{2} + (\alpha- 1) } - I(X;Y),
        \label{EQ:optimized_epsilon} \\
        \beta_2 &= \beta.
    \end{align}
    This gives the desired result.
\end{proof}

\begin{proof}[Proof of Theorem~\ref{theorem second order}]
    For the analysis of the second-order rate, we use many of the same steps as the proof of Theorem~\ref{theorem exponential}.  Assume all of the same definitions.
    
    The key difference is the bound on $\mathbb{P}_Q \left( \overline{{\cal A}_{\epsilon}} \right)$ found in \eqref{atypical probability bound}.  Instead of using the Chernoff bound we will use the Berry-Esseen theorem.
    \begin{align}
        \mathbb{P}_Q \left( \overline{{\cal A}_{\epsilon}} \right) &\leq {\cal Q} \left( \frac{\epsilon \sqrt{n}}{\sqrt{V}} \right) + \frac{\rho}{V^{3/2} \sqrt{n}},
    \end{align}
    where $\rho = \mathbb{E} | \imath_{X;Y}(X;Y) - I(X;Y) |^3 \leq \infty$ because ${\cal X}$ and ${\cal Y}$ are finite.
    
    Now choose $r \in (0, c - d - 1)$, where $c$ and $d$ are from the theorem statement, and let
    \begin{equation}
        \epsilon = \frac{1}{\sqrt{n}} {\cal Q}^{-1}(\varepsilon) \sqrt{V} + r \frac{ \log n }{n}.
    \end{equation}

    The bound on $\mathbb{P}_Q \left( \overline{{\cal A}_{\epsilon}} \right)$ becomes
    \begin{equation}
        \mu_n \triangleq {\cal Q} \left( {\cal Q}(\varepsilon) + \frac{r}{\sqrt{V}} \frac{\log n}{\sqrt{n}} \right) + \frac{\rho}{V^{3/2} \sqrt{n}}.
    \end{equation}
    
    Now, in a step analogous to \eqref{chernoff for atypical}, again using the Chernoff bound,
    \begin{align}
        \mathbb{P} \left( \int d P_{{\cal C}, 2} \geq \mu_n \left( 1 + \frac{1}{\sqrt{n}} \right) \right) &\leq e^{-\frac{\mu_n}{3n} 2^{nR}}.
    \end{align}
    
    Also, in a step analogous to \eqref{typical set bound}, we use the Chernoff bound to obtain the following:
    \begin{align}
        \mathbb{P} \left( D_{{\cal C},1}(y^n) \geq 1 + \frac{1}{\sqrt{n}} \right) &\leq e^{-\frac{1}{3n} 2^{n( R - I(X;Y) - \epsilon )}} \quad \forall v^n \\
        &= e^{-\frac{1}{3n} 2^{n \left( (c-r) \frac{\log n}{n} \right)}} \\
        &= e^{-\frac{1}{3} n^{c-r-1}}.
    \end{align}

    At this point we will use the fact that ${\cal Y}$ is a finite set to apply the union bound.  Let $\bar{\cal S}$ be the set of codebooks such that the following are true:
    \begin{align}
    \int d P_{{\cal C},2} &< \mu_n \left( 1 + \frac{1}{\sqrt{n}} \right), \\
    D_{{\cal C}, 1}(y^n) &< 1 + \frac{1}{\sqrt{n}} \quad \forall y^n \in {\cal Y}^n.
    \end{align}
    The probability of not being in $\bar{\cal S}$ is super-exponentially small:
    \begin{align}
    \mathbb{P}({\cal C} \notin \bar{\cal S}) &\leq e^{-\frac{\mu_n}{3n} 2^{nR}} + |{\cal Y}|^n e^{-\frac{1}{3} n^{c-r-1}}.
    \end{align}
    Notice that $\mu_n$ converges to $\varepsilon$, so the above probability is dominated by the second term.  Since $d < c-r-1$, this establishes the probability statement of the theorem.
    
    What remains is to show that for every codebook in $\bar{\cal S}$, the total variation is eventually less than $\varepsilon$.  From \eqref{expanded total variation bound} it follows that
    \begin{align}
        \| P_{Y^n|{\cal C}} - Q_{Y^n}\|_{TV} &\leq \mu_n \left( 1 + \frac{1}{\sqrt{n}} \right) + \frac{1}{\sqrt{n}}.
    \end{align}
    As observed previously, $\mu_n$ converges to $\varepsilon$.  Furthermore, it converges from below and the deviation is of order $\frac{\log n}{\sqrt{n}}$, which dominates the other terms in the total variation bound.
\end{proof}

\section{Applications}
\label{section applications}

As stated in \cite{cuff15}, these stronger versions of Wyner's soft-covering lemma have important applications, particularly to information theoretic security.  The main advantage of these theorems come from the union bound.

The usual random coding argument for information theory uses a randomly generated codebook until the final steps of the achievability proof.  In these final steps, it is claimed that there exists a good codebook based on the analysis.  This can be done by analyzing the expected value of the performance for the random ensamble and claiming that at least one codebook is as good as the expected value.  Alternatively, one can make the argument based on the probability that the randomly generated codebook has a good performance.  If that probability is greater than zero, then there is at least one good codebook.  The second approach can be advantageous when performance is not captured by one scalar value that is easily analyzed---for example, if ``good'' performance involves a collection of constraints.

These stronger soft-covering theorems give a very strong assurance that soft-covering will hold.  Even if the codebook needs to satisfy exponentially many constraints related to soft-covering, the union bound will yield the claim that a codebook exists which satisfies them all simultaneously.  Indeed, if you ran the soft-covering experiment exponentially many times, regardless of how the codebooks are correlated from one experiment to the next, the probability of seeing even one fail is still super-exponentially small.

We have shown in \cite{goldfeld-cuff-permuter15} that high probability theorems for soft covering can be used to tackle previously challenging problems in secure communication where precisely this need arises.  That work demonstrates that in the wiretap channel of type II \cite{ozarow-wyner84}, where an adversary can influence the channel with an exponential number of possible actions, a random codebook will achieve secrecy for all of them simultaneously.  Furthermore, ``semantic security,'' which is a very practical notion of secrecy but is stronger than the secrecy typically guaranteed in information theory, requires security to hold even for the most distinguishable pair of messages.  This level of secrecy is shown to be achieved, again using the union bound and the super-exponential assurance of soft covering.

\section*{Acknowledgment}

This work was supported by the National Science Foundation (grant CCF-1350595) and the Air Force Office of Scientific Research (grant FA9550-15-1-0180).

\end{document}